\documentclass[
  aps,
  pra,
  reprint,
  superscriptaddress,
  amsmath,
  amssymb,
  nofootinbib,
  longbibliography
]{revtex4-2}

\usepackage[T1]{fontenc}
\usepackage[american]{babel}
\usepackage{mathtools}
\usepackage{amsthm}
\usepackage{graphicx}
\usepackage{xcolor}
\usepackage{booktabs}
\usepackage{bm}
\usepackage{braket}
\usepackage{microtype}
\usepackage{hyperref}
\usepackage[nameinlink,capitalise,noabbrev]{cleveref}
\usepackage{orcidlink}
\usepackage{cleveref}
\definecolor{2e308d}{RGB}{46,48,141}

\RequirePackage{tikz, pgfplots}
\usetikzlibrary{arrows.meta}
\usetikzlibrary{backgrounds}
\usepgfplotslibrary{groupplots}
\definecolor{matlabBlue}{rgb}{0,0.4470,0.7410}
\definecolor{matlabOrange}{rgb}{0.8500,0.3250,0.0980}

\newcommand{\qinc}{Quantum Innovation Centre (Q.InC), \href{https://ror.org/036wvzt09}{Agency for Science, Technology and Research (A*STAR)}, 4 Fusionopolis Way, Kinesis \#05-01, Singapore 138635, Republic of Singapore\looseness=-1}

\newcommand{\smt}{Science, Mathematics and Technology Cluster, \href{https://ror.org/05j6fvn87}{Singapore University of Technology and Design}, 8 Somapah Road, Singapore 487372, Republic of Singapore\looseness=-1}

\newcommand{\esd}{Engineering Systems Design Pillar, \href{https://ror.org/05j6fvn87}{Singapore University of Technology and Design}, 8 Somapah Road, Singapore 487372, Republic of Singapore\looseness=-1}

\newcommand{\sit}{Engineering Cluster, \href{https://ror.org/01v2c2791}{Singapore Institute of Technology}, 1 Punggol Coast Road, Singapore 828608, Republic of Singapore\looseness=-1}

\hypersetup{
  colorlinks=true,
  linkcolor=2e308d,
  citecolor=2e308d,
  urlcolor=2e308d
}

\graphicspath{{figures/}{../figures/}{scripts/output/figures/}}
\DeclareGraphicsExtensions{.pdf,.png,.jpg}

\newtheorem{theorem}{Theorem}
\newtheorem{lemma}[theorem]{Lemma}
\newtheorem{proposition}[theorem]{Proposition}

\theoremstyle{definition}

\theoremstyle{remark}

\crefname{claim}{claim}{claims}
\Crefname{claim}{Claim}{Claims}

\IfFileExists{macros/operators.tex}{%

\DeclarePairedDelimiterX{\oointerval}[2]{\lparen}{\rparen}{#1,\,#2} %
\DeclarePairedDelimiterX{\ocinterval}[2]{\lparen}{\rbrack}{#1,\,#2} %
\DeclarePairedDelimiterX{\cointerval}[2]{\lbrack}{\rparen}{#1,\,#2} %
\DeclarePairedDelimiterX{\ccinterval}[2]{\lbrack}{\rbrack}{#1,\,#2} %

\DeclarePairedDelimiterX{\meanbrackets}[1]{\lbrack}{\rbrack}{%
   #1%
}
\DeclareMathOperator*{\meansymbol}{\mathbb{E}}
\NewDocumentCommand{\mean}{e{_}}{%
  \meansymbol%
  \IfNoValueF{#1}{\sb{#1}}%
  \meanbrackets%
} %

\DeclarePairedDelimiterX{\probparen}[1]{\lparen}{\rparen}{%
   #1%
}
\DeclareMathOperator*{\probsymbol}{Pr}
\NewDocumentCommand{\prob}{e{_}}{%
  \probsymbol%
  \IfNoValueF{#1}{\sb{#1}}%
  \probparen%
} %

\DeclareMathOperator{\simplex}{\Delta} %

\newcommand{\amatrix}[1]{\mathbf{#1}} %

\DeclarePairedDelimiterX{\inner}[2]{\langle}{\rangle}{#1,\,#2} %
}{}
\IfFileExists{macros/constants.tex}{%

\providecommand{%
  \naturals}{\mathbb{N}}            %

\providecommand{%
  \integers}{\mathbb{Z}}            %

\providecommand{%
  \nonnegativeintegers
}{%
  \integers_{\geq 0}}               %

\providecommand{%
  \rationals}{\mathbb{Q}}           %

\providecommand{%
  \reals}{\mathbb{R}}               %

\providecommand{%
  \simplex}{\Delta}                 %

\providecommand{%
  \nonnegativereals
}{%
  \reals_{\geq 0}}                  %

\providecommand{%
  \complex}{\mathbb{C}}             %

\providecommand{%
  \complextorus}{\mathbb{T}}        %

\providecommand{%
  \differentiable}{\mathcal{C}}     %

\providecommand{%
  \polynomials}{\reals}             %

\providecommand{%
  \apolynomial}{p}                  %

\providecommand{%
  \sumofsquares}{\mathbbm{\Sigma}}  %

\providecommand{%
  \quadraticmodule}{Q}              %

\providecommand{%
  \asumofsquares}{\sigma}           %

\providecommand{%
  \unitvector}{e}                   %

\providecommand{%
  \identitymatrix}{\amatrix{I}}     %

\providecommand{%
  \hessian}{\amatrix{H}}            %

\providecommand{%
  \jacobian}{\amatrix{J}}           %

\providecommand{%
  \symmetrizedjacobian%
}{\amatrix{SJ}}                     %

\providecommand{%
  \mineigenvalue%
}{%
  \lambda_{\mathrm{min}}%
}                                   %

\providecommand{%
  \maxeigenvalue%
}{%
  \lambda_{\mathrm{max}}%
}                                   %

\providecommand{%
  \im}{i}                           %
  
\providecommand{%
  \normaldistribution%
}{\ensuremath{\mathcal{N}}}         %
}{}

\pgfplotsset{compat=1.18}

\begin{document}

\title{Training Variational Quantum Algorithms Is NP-Hard, Even Locally}

\author{Dax Enshan Koh\,\orcidlink{0000-0002-8968-591X}}
\email{dax.koh@singaporetech.edu.sg}
\affiliation{\sit}
\affiliation{\smt}
\affiliation{\qinc}

\author{Triscia Mundo\,\orcidlink{0009-0009-1332-5860}}
\email{triscia\_mundo@mymail.sutd.edu.sg}
\affiliation{\esd}

\author{Iosif Sakos\,\orcidlink{0000-0002-1871-9078}}
\email{iosif\_sakos@sutd.edu.sg}
\affiliation{\esd}

\author{Antonios Varvitsiotis\,\orcidlink{0000-0002-9658-6469}}
\email{antonios@sutd.edu.sg}
\affiliation{\esd}

\affiliation{\href{https://ror.org/01mgdzc49}{Centre for Quantum Technologies}, \href{https://ror.org/01tgyzw49}{National University of Singapore}, Singapore}
\affiliation{Archimedes Research Unit on AI, Data Science and Algorithms, Athena RC, Marousi, Greece}

\begin{abstract}                      

Variational quantum algorithms (VQAs) generally rely on classical optimization to train parameterized quantum circuits. This training seeks to minimize an objective function, and its efficiency is central to the practical success of these algorithms. However, globally minimizing such training objectives over the circuit parameters is known to be $\mathsf{NP}$-hard, limiting the prospect of general guarantees for efficient training. In this Letter, we prove that even the weaker task of finding a local minimum of such VQA training objectives is strongly $\mathsf{NP}$-hard, including when the objective admits efficient classical evaluation. Moreover, we show that this hardness persists even for the task of finding a parameter vector within $\ell_p$-distance strictly less than $\pi/2$ of some local minimizer, for every $p\geq 1$. Our central technical result is that approximating a local minimizer of a Hermitian trigonometric polynomial is strongly $\mathsf{NP}$-hard. By explicitly constructing quantum circuits whose training objectives reproduce these hard instances, we obtain a polynomial-time reduction to VQA training. Our results establish a fundamental computational barrier to variational quantum training: even reaching the vicinity of a local minimum remains hard in the worst case.

\end{abstract}

\maketitle

\section{Introduction}

Quantum algorithms promise to transform our ability to solve problems beyond the practical reach of classical computers \cite{dalzell2025quantum,montanaro2016quantum}. Realizing this promise, however, requires accounting for their full computational cost, including the resources consumed by their classical components \cite{dalzell2025quantum}. This issue is particularly important for variational quantum algorithms (VQAs), which generally use classical optimization to train parameterized quantum circuits \cite{cerezo2021variational,mcClean2016theory,bharti2022noisy,wang2026review}. This variational framework underlies a myriad of applications, including ground-state energy estimation \cite{peruzzo2014variational,tilly2022variational,you2026exploring}, combinatorial optimization \cite{farhi2014quantum,herrman2022multi-angle,vijendran2024expressive,blekos2024review}, solving differential equations \cite{lubasch2020variational,leong2022variational,sarma2024quantum,leong2024solving}, and quantum machine learning \cite{havlicek2019supervised,mitarai2018quantum,wu2022challenges,wan2017quantum}. Across these applications, a costly classical search for suitable circuit parameters can offset potential computational gains from quantum processing \cite{bittel2021training,anschuetz2022quantum,anschuetz2022critical,bittel2023optimal}.

In a common formulation of VQA training, the circuit parameters $\theta$ are chosen to minimize the training objective $F(\theta)=\langle 0| U(\theta)^\dag HU(\theta)|0\rangle$, where $U(\theta)$ is the unitary operator corresponding to a parameterized quantum circuit and $H$ is a Hermitian observable encoding the task \cite{mcClean2016theory,cerezo2021variational}. Two distinct difficulties can obstruct efficient training. The first is the measurement cost of estimating objective values or gradients of $F$ \cite{wecker2015progress,wang2021minimizing,straaten2021measurement,gonthier2022measurements}. This cost can be particularly severe in the presence of barren plateaus, where resolving gradients can require exponentially many measurements as the system size grows \cite{mcclean2018barren,cerezo2021cost,arrasmith2021effect,arrasmith2022equivalence,zhang2024exponential,larocca2025barren}. The second is the computational difficulty of searching the parameter space. 
Globally minimizing $F$ over $\theta$ is known to be $\mathsf{NP}$-hard in the worst case, even when $F$ admits efficient classical evaluation \cite{bittel2021training}, posing a fundamental obstacle to general guarantees of efficient training.

In practice, VQA training often seeks to improve the objective through successive parameter updates guided by local information about $F$ \cite{schuld2019evaluating,stokes2020quantum,sweke2020stochastic,harrow2021low}. Local minimizers are natural targets for such methods \cite{anschuetz2022quantum,wierichs2020avoiding}: they are parameter vectors for which no sufficiently small change lowers the objective. Seeking a local minimizer instead of a global one could, in principle, make training more tractable \cite{cartis2022evaluation}. Crucially, the $\mathsf{NP}$-hardness of global optimization does not, ipso facto, establish the hardness of this weaker task, since a landscape may admit local minima that are easy to find even when finding a global minimum is hard. This raises a basic question for VQA training: can the computational barrier be avoided if the goal is only to find a local minimizer, or even a parameter vector close to one?

In this Letter, we answer this question in the negative. Specifically, we prove that finding a local minimizer of a VQA training objective is strongly $\mathsf{NP}$-hard, even when the objective admits efficient classical evaluation. Here, strong $\mathsf{NP}$-hardness means that the problem remains $\mathsf{NP}$-hard even when the numerical quantities specifying the instance are polynomially bounded in the problem size \cite{garey1978strong}. Moreover, we show that this hardness persists even when the goal is merely to find a parameter vector within $\ell_p$-distance strictly less than $\pi/2$ of some local minimizer, for every $p \geq 1$. Notably, this hardness already holds for circuits consisting of single-qubit rotations and a commuting diagonal two-local observable. Our results therefore show that relaxing global optimality does not eliminate the worst-case computational barrier to VQA training: even reaching the vicinity of a local minimizer remains hard, independently of measurement uncertainty.

The key technical ingredient in our proof is the strong $\mathsf{NP}$-hardness of approximating a local minimizer of a Hermitian trigonometric polynomial, even when the total Fourier degree is at most two. We establish this hardness result for Hermitian trigonometric polynomials in two steps. First, we obtain a polynomial-time reduction from \textsc{Stable Set} to the local minimization of a quadratic function $q$ over $[0,1]^n$. For a graph $G$ and positive integer threshold $k$, the construction exhibits a sharp dichotomy: if the stability number $\alpha(G)<k$, then the origin is the unique local minimizer, whereas if $\alpha(G)\ge k$, then every local minimizer has at least one coordinate equal to $1$. This is a bounded-box analogue of the dichotomy originally shown by Ahmadi and Zhang~\cite{ahmadi2022complexity} for quadratic minimization over a scaled simplex. Second, we apply a trigonometric substitution that converts the quadratic objective into a Hermitian trigonometric polynomial while preserving local minimality. The dichotomy then translates into a separation between the possible locations of local minimizers in parameter space, yielding the claimed hardness of approximation; see \cref{fig:local_minimizer_separation}.

To establish the corresponding hardness for VQA training, we use the connection between trigonometric polynomials and expectation values of parameterized quantum circuits~\cite{you2021exponentially,schuld2021effect,wierichs2022general,nemkov2023fourier,sakos2026globaloptpqc}. For each hard trigonometric polynomial $f$, we explicitly construct a polynomial-size circuit and observable whose training objective reproduces it exactly. The resulting identity $F(\theta) = f(\theta)$ preserves the local minimizers and distances to them, completing the polynomial-time reduction to VQA training.

\definecolor{mapoutline}{HTML}{293B49}
\definecolor{mapfill}{HTML}{EAF3F9}
\definecolor{mappoint}{HTML}{1C55AE}

\section{Main result}

We begin with a few definitions and conventions for representing the optimization instances. A \textit{Hermitian trigonometric polynomial} in $n$ variables is a function of the form
\begin{align}\label{HTP}
    f(\theta)
    =
    \sum_{\alpha\in\mathcal{A}}
    c_\alpha \mathrm{e}^{\mathrm{i}\alpha^{\mathsf{T}}\theta}, \qquad
    \theta \in \mathbb R^n,
\end{align}
where $\mathcal A \subset \mathbb Z^n$ is a finite symmetric set ($\mathcal A = -\mathcal A$) and $c_{-\alpha}=\overline{c_\alpha}$ for every $\alpha \in \mathcal A$ \cite{dumitrescu2017multivariate}. Its total \textit{Fourier degree} is
$\mathrm{deg}_F(f) = \max_{
\alpha \in \mathcal A: c_\alpha\neq 0} \|\alpha\|_1 
$, with $\mathrm{deg}_F(0)=0$.

For our complexity statements, all instances are specified explicitly using integers and rational numbers encoded in binary; the input size is the total description length.

Our main result on the computational hardness of VQA training is as follows:

\begin{theorem}\label{thm:vqa-main}
Finding a local minimizer of an efficiently represented VQA training objective $F(\theta)$ is strongly $\mathsf{NP}$-hard. More precisely, for every $p \ge 1$, it is strongly $\mathsf{NP}$-hard to find a parameter vector in $[0,2\pi]^n$ within $\ell_p$-distance strictly less than $\pi/2$ of some local minimizer.
\end{theorem}

The proof of \cref{thm:vqa-main} relies on the following hardness result for Hermitian trigonometric polynomials.

\begin{proposition}\label{thm:htp-nphard}
For every $p \geq 1$, approximating a local minimizer of an efficiently represented Hermitian trigonometric polynomial of total Fourier degree at most $2$ within $\ell_p$-distance strictly less than $\pi / 2$ is strongly $\mathsf{NP}$-hard.
\end{proposition}

The proof of \cref{thm:htp-nphard} relies on a reduction from the \textsc{Stable Set} problem.
This is an $\mathsf{NP}$-complete problem \cite{karp1972reducibility} in which, given a graph $G$ and a positive integer $k$, one must determine whether there exists a stable set of $G$, i.e., a set of pairwise non-adjacent vertices, of size at least $k$.
The size of the largest such set is called the \emph{stability number} of $G$, and is denoted by $\alpha(G)$.

The first step of the proof is to establish the dichotomy of local minimizers for the quadratic form associated with the \textsc{Stable Set} problem, as stated in \cref{lemma:quadratic-bounded}.
This lemma is a modification of the result by \citet{ahmadi2022complexity}. %

\begin{figure}
  \centering
  \begin{tikzpicture}[
    font=\small,
    every node/.append style={font=\small}
  ]
    \def\plotpadding{0.025}

    \tikzset{
      minimizer outline/.style={
        circle,
        minimum size=6pt,
        inner sep=0pt,
        fill=white,
        draw=white
      },
      minimizer point/.style={
        minimizer outline,
        minimum size=4pt,
        fill=black,
        draw=black,
        line width=0.4pt
      },
      pics/minimizer/.style args={#1}{
        code={
          \node[minimizer outline] {};
          \node[minimizer point, fill=#1] {};
        }
      },
      pics/minimizer/.default=black,
      minimizer set/.style args={#1}{
        /pgfplots/on layer=axis background,
        draw=#1,
        opacity=0.8,
        line width=2.5pt,
        line cap=butt,
        line join=miter
      },
      minimizer set/.default=black
    }

    \pgfplotsset{
      layers/with background/.define layer set={
        background,
        axis background,
        axis grid,
        axis ticks,
        axis lines,
        axis tick labels,
        pre main,
        main,
        axis descriptions,
        axis foreground
      }{
        /pgfplots/layers/standard
      },
      transformed box/.style={
        xmin=0,
        xmax={2*pi},
        ymin=0,
        ymax={2*pi},
        xtick={0,pi,2*pi},
        xticklabels={$0$,$\pi$,$2\pi$},
        ytick={0,pi,2*pi},
        yticklabels={$0$,$\pi$,$2\pi$}
      },
      box/.style={
        xmin=0,
        xmax=1,
        ymin=0,
        ymax=1,
        xtick={0,1},
        ytick={0,1}
      }
    }

    \begin{groupplot}[
      group style={
        group size=2 by 2,
        horizontal sep=0.10\columnwidth,
        vertical sep=9mm
      },
      width=0.34\columnwidth,
      height=0.34\columnwidth,
      scale only axis,
      tick label style={font=\small},
      set layers=with background,
      axis lines=middle,
      axis line style={-{Stealth}},
      axis equal image,
      enlarge x limits=\plotpadding,
      enlarge y limits=\plotpadding
    ]

      \nextgroupplot[transformed box]

      \pic at (axis cs:0,0) {minimizer=matlabOrange};
      \pic at (axis cs:0,{2*pi}) {minimizer=matlabOrange};
      \pic at (axis cs:{2*pi},0) {minimizer=matlabOrange};
      \pic at (axis cs:{2*pi},{2*pi}) {minimizer=matlabOrange};

      \node (preimage) at (axis cs:{pi},{pi})
        {$T^{-1}(x^{\star})$};

      \draw[-{Stealth}, semithick, opacity=0.5, shorten >=3pt]
        (preimage) -- (axis cs:0,0);
      \draw[-{Stealth}, semithick, opacity=0.5, shorten >=3pt]
        (preimage) -- (axis cs:0,{2*pi});
      \draw[-{Stealth}, semithick, opacity=0.5, shorten >=3pt]
        (preimage) -- (axis cs:{2*pi},0);
      \draw[-{Stealth}, semithick, opacity=0.5, shorten >=3pt]
        (preimage) -- (axis cs:{2*pi},{2*pi});

      \nextgroupplot[
        transformed box,
        legend to name=minimizer-legend,
        legend columns=1,
        legend style={
          font=\small,
          draw=none,
          cells={anchor=west},
          /tikz/row sep=1pt
        }
      ]

      \draw[minimizer set=matlabOrange]
        (axis cs:0,{pi}) -- (axis cs:{2*pi},{pi});
      \draw[minimizer set=matlabOrange]
        (axis cs:{pi},0) -- (axis cs:{pi},{2*pi});

      \addlegendimage{
        only marks,
        mark=*,
        mark size=2pt,
        mark options={
          fill=matlabOrange,
          draw=black,
          line width=0.4pt
        }
      }
      \addlegendentry{Possible local minimizers in $\Theta$}

      \addlegendimage{
        only marks,
        mark=*,
        mark size=2pt,
        mark options={
          fill=matlabBlue,
          draw=black,
          line width=0.4pt
        }
      }
      \addlegendentry{Possible local minimizers in $\mathcal{X}$}

      \nextgroupplot[box]

      \pic at (axis cs:0,0) {minimizer=matlabBlue};
      \node[
        anchor=south west,
        xshift=4pt,
        yshift=4pt,
        inner sep=0pt
      ] at (axis cs:0,0) {$x^{\star}=0$};

      \nextgroupplot[box]

      \draw[minimizer set=matlabBlue]
        (axis cs:0,1) -- (axis cs:1,1) -- (axis cs:1,0);

    \end{groupplot}

    \node[text=matlabOrange, xshift=-3mm]
      (theta-space) at (group c1r1.left of west)
      {$\Theta$};

    \node[text=matlabBlue]
      (x-space) at (theta-space |- group c1r2.left of west)
      {$\mathcal{X}$};

    \draw[-{Stealth}, thick]
      (theta-space) -- node[left] {$T$} (x-space);

    \node[anchor=south]
      (no-instance) at (group c1r1.above north)
      {No-Instance ($\alpha(G)<k$)};

    \node[anchor=south]
      (yes-instance) at (group c2r1.above north)
      {Yes-Instance ($\alpha(G)\geq k$)};

    \node[anchor=north, yshift=-3mm]
      at ($(group c1r2.below south)!0.5!(group c2r2.below south)$)
      {\pgfplotslegendfromname{minimizer-legend}};

  \end{tikzpicture}

  \caption{Illustration of the local minimizer separation in
  \cref{lemma:quadratic-bounded}, shown in $n=2$ dimensions.
  For $\alpha(G)<k$, the local minimizer $x^\star=0$ lifts to
  $T^{-1}(0)=\{0,2\pi\}^2$. For $\alpha(G)\geq k$, the local
  minimizers lie on the highlighted areas. Here, $\Theta = [0,2\pi]^2$ and $\mathcal X=[0,1]^2$.}
  \label{fig:local_minimizer_separation}
\end{figure}
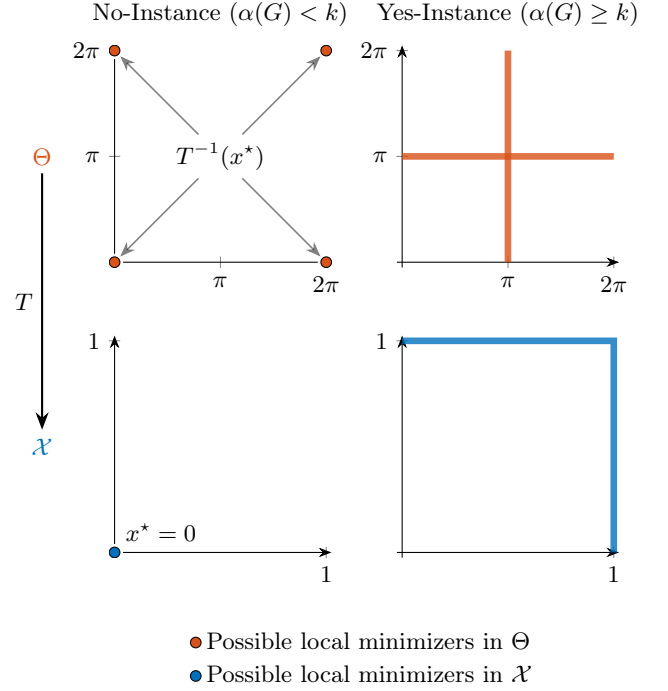

\begin{lemma}\label{lemma:quadratic-bounded}
Let $G$ be a graph on $n$ vertices with stability number $\alpha(G)$, and let $k$ be a positive integer satisfying $1\le k\le n$. In time polynomial in $n$, we can construct a quadratic function
$q(x)
= x^{\mathsf{T}} M x$,
where $M$ is a matrix that depends on $G$ and $k$,
for which 
the local minimizers of $q$ over $[0,1]^n$ satisfy the following dichotomy:
\begin{enumerate}
\item If $\alpha(G) < k$, then $x^{\star} = 0$ is the unique local minimizer.
\item If $\alpha(G)\ge  k$, then every local minimizer $x^{\star}$ has at least one component equal to $1$.
\end{enumerate}
\end{lemma}
\begin{proof}
Let $A$ be the adjacency matrix of $G$, and let $I$ and $J$ denote the identity matrix and the all-ones matrix, respectively; the matrix in the construction is $M=\left(k-\frac{1}{2}\right)(A+I)-J$.
Suppose first that the stability number $\alpha(G)$ is less than $k$.
In this case, by \citet[Corollary~2.3]{ahmadi2022complexity}, we have that $q$ is strictly copositive in $\mathbb{R}_{\geq 0}^n$, i.e., $q(x) > 0$ for all nonzero $x \in \mathbb{R}_{\geq 0}^n$ (and $q(0) = 0$).
Therefore, $x = 0$ is the unique global minimizer of $q$ over $\mathbb{R}_{\geq 0}^n$.
However, this does not immediately imply that $x = 0$ is the unique local minimizer, which we establish next.

Consider a nonzero point $x \in [0,1]^n$, and let $u = 0-x$ be the direction from $x$ towards $0$.
Since $q$ is a homogeneous quadratic, it follows that for every $\lambda \in (0,1)$,
\begin{align}
q(x+\lambda u)
= q\bigl((1-\lambda)x\bigr) 
= (1-\lambda)^2q(x) 
< q(x),
\end{align}
where the last inequality follows from strict copositivity. Thus, the direction $u$ is a descent direction for $q$ at $x$, and therefore no nonzero $x$ can be a local minimizer.
Consequently, $x = 0$ is the unique local minimizer of $q$ over $[0,1]^n$.

Next, suppose instead that the stability number $\alpha(G)$ is greater than or equal to $k$.
Then, since $[0,1]^n$ is compact and $q$ is continuous, $q$ has at least one minimizer over $[0,1]^n$.
Moreover, by \citet[proof of Theorem~2.5]{ahmadi2022complexity}, we have that $q$ has no local minimizer in $\mathbb{R}_{\geq 0}^n$.
Now, let $x \in [0,1]^n$ have all its components strictly less than $1$.
Then, $[0,1]^n$ contains an open neighborhood around $x$ relative to $\mathbb{R}_{\geq 0}^n$.
Thus, if $x$ were a local minimizer of $q$ over $[0,1]^n$, it would also be a local minimizer in $\mathbb{R}_{\geq 0}^n$, which is a contradiction.
Consequently, every local minimizer of $q$ over $[0,1]^n$ must have at least one component equal to $1$.
\end{proof}

Then, the following lemma shows that the dichotomy established in \cref{lemma:quadratic-bounded} is preserved under continuous open surjections, and is the main machinery used in our reduction; see \cref{fig:local-homeomorphism} for an illustration.

\begin{lemma}\label{lemma:local-homeomorphism}
Let $T:{\Theta} \to \mathcal{X}$ be a continuous open surjection, i.e., a continuous surjective map that maps open sets in $\Theta$ to open sets in $\mathcal{X}$.
Let $f:\mathcal{X} \to \mathbb{R}$, and define
\begin{align*}
F(\theta)
= f\bigl(T(\theta)\bigr).
\end{align*}
Then, for every $\theta \in \Theta$ and $x = T(\theta)$, $x$ is a local minimizer of $f$ over $\mathcal{X}$ if and only if $\theta$ is a local minimizer of $F$ over~$\Theta$.
\end{lemma}
\begin{proof}
Suppose $x = T(\theta)$ is a local minimizer of $f$ over $\mathcal{X}$.
Then there is a neighborhood $\mathcal{U}$ of $x$ such that $f(x) \leq f(x')$ for every $x' \in \mathcal{U}$.
Moreover, since $T$ is continuous, the preimage $T^{-1}(\mathcal{U})$ is a neighborhood of $\theta$.
Then, for every $\theta' \in T^{-1}(\mathcal{U})$, we have
\begin{equation}
F(\theta')
= f\bigl(T(\theta')\bigr)
\geq f(x)
= F(\theta).
\end{equation}
Hence, $\theta$ is a local minimizer of $F$ over $\Theta$.

Conversely, suppose $\theta$ is a local minimizer of $F$ over $\Theta$.
Then there exists an open (with respect to $\Theta$) neighborhood $\mathcal{V}$ of $\theta$ such that $F(\theta) \leq F(\theta')$ for every $\theta' \in \mathcal{V}$.
Since $T$ is an open surjection, $T(\mathcal{V})$ is an open neighborhood of $x = T(\theta)$ in $\mathcal{X}$.
Thus, for every $x' \in T(\mathcal{V})$, there exists $\theta' \in \mathcal{V}$ such that $x' = T(\theta')$, and we have
\begin{equation}
f(x')
= f\bigl(T(\theta')\bigr)
= F(\theta')
\geq F(\theta)
= f(x).
\end{equation}
Hence, $x$ is a local minimizer of $f$ over $\mathcal{X}$.
\end{proof}

We are now ready to present the main reduction.

\begin{proof}[Proof of \cref{thm:htp-nphard}]
Consider an instance of the \textsc{Stable Set} problem given by a graph $G$ and an integer $k \in \{1,2,\ldots,n\}$, where $n$ is the number of vertices of $G$.

Next, define the map
\begin{align}
T:[0,2\pi] \to [0,1],
\qquad
T(\theta)
= \frac{1 - \cos(\theta)}{2},
\end{align}
and abusing notation slightly, we also write $T(\theta) = \bigl(T(\theta_1),\dots,T(\theta_n)\bigr)$ for $\theta \in [0,2\pi]^n$.
The map $\theta \mapsto \cos(\theta)$ is a continuous open surjection from $[0,2\pi]$ to $[-1,1]$, and therefore so are $T$ and its componentwise extension from $[0,2\pi]^n$ to $[0,1]^n$.
This follows because these properties are preserved under composition with nonzero affine maps and under Cartesian products.

We then construct a Hermitian trigonometric polynomial $f$ by performing a change of variables on $q$ via the map $T$, i.e.,
\begin{align}\label{eq:HTPChangeOfVariables}
f(\theta)
= q\bigl(T(\theta)\bigr)
= \sum_{i = 1}^{n} \sum_{j = 1}^{n} M_{ij} T(\theta_i) T(\theta_j).
\end{align}
Observe that $f$ is indeed a Hermitian trigonometric polynomial since
\begin{align}
T(\theta_i)
= \frac{1 - \cos(\theta_i)}{2}
= \frac{1}{2} - \frac{1}{4}\bigl(\mathrm{e}^{\mathrm{i}\theta_i} + \mathrm{e}^{-\mathrm{i}\theta_i}\bigr),
\end{align}
and therefore by expanding the terms $T(\theta_i)T(\theta_j)$ in $q\bigl(T(\theta)\bigr)$ and collecting identical frequencies, we have
\begin{align}
f(\theta)
= \sum_{\alpha \in \mathcal{A}} c_{\alpha} \mathrm{e}^{\mathrm{i}\alpha^{\mathsf{T}}\theta},
\end{align}
where $\mathcal{A} = \{\alpha \in \mathbb{Z}^n : \lVert\alpha\rVert_1 \leq 2\}$, and the coefficients $c_{\alpha}$ are linear combinations of the entries of the matrix $M$, and satisfy $c_{-\alpha} = c_{\alpha} = \overline{c_{\alpha}}$.
Thus, $f$ is a Hermitian trigonometric polynomial, and both its frequencies $\alpha$ and coefficients $c_{\alpha}$ can be constructed in polynomial time.

We now use \cref{lemma:local-homeomorphism} to transfer the dichotomy of local minimizers from $q$ over $[0,1]^n$ to $f$ over $[0,2\pi]^n$.
In particular, if the stability number $\alpha(G) < k$, then $0$ is the unique local minimizer of $q$, and therefore by \cref{lemma:local-homeomorphism}, $\theta$ is a local minimizer of $f$ if and only if $T(\theta) = 0$, i.e., for all coordinates $i$,
\begin{align}
\frac{1 - \cos(\theta_i)}{2}
= 0
\implies
\theta_i \in \{0,2\pi\}.
\end{align}
On the other hand, if the stability number $\alpha(G) \ge  k$, then every local minimizer $x^{\star}$ of $q$ has at least one component equal to $1$, and therefore by \cref{lemma:local-homeomorphism}, every local minimizer $\theta$ of $f$ satisfies $T(\theta_i)= 1$ for at least one $i$, i.e.,
\begin{align}
\frac{1 - \cos(\theta_i)}{2}
= 1
\implies
\theta_i = \pi.
\end{align}

Next, we translate this dichotomy into an inapproximability result.
For this, suppose there exists a polynomial-time algorithm that returns a point $\theta \in [0,2\pi]^n$ within $\ell_{\infty}$-distance strictly less than $\pi / 2$ of some local minimizer $\theta^{\star}$ of $f$.
If $\alpha(G) < k$, then $\theta^{\star} \in \{0,2\pi\}^n$, and therefore for each coordinate $i$,  we have 
\begin{align}
\lvert\theta_i\rvert
< \frac{\pi}{2}\  \text{ or }\  \lvert\theta_i\rvert
> \frac{3\pi}{2}. 
\end{align}
On the other hand, if $\alpha(G)\ge k$, then there exists a coordinate $i$ such that $\theta_i^{\star} = \pi$, and therefore $\lvert\theta_i - \pi\rvert < \frac{\pi}{2}$, which implies
\begin{align}
\frac{\pi}{2}
< \lvert\theta_i\rvert
< \frac{3\pi}{2}.
\end{align}
Therefore, we can distinguish between the cases $\alpha(G) < k$ and $\alpha(G) \ge k$ by examining the magnitude of the components of $\theta$.
This can be performed in polynomial time, and thus the reduction is complete.

Finally, as $\lVert\theta - \theta^{\star}\rVert_{\infty} \leq \lVert\theta - \theta^{\star}\rVert_p$ for all $p \geq 1$, any approximation within $\ell_p$-distance strictly less than $\pi / 2$ also guarantees an approximation within $\ell_{\infty}$-distance strictly less than $\pi / 2$ for all $p \geq 1$.
This extends the inapproximability result to all $\ell_p$-norms for $p \geq 1$. Since $1\le k \le n$, the entries of $M$ have magnitude $O(n)$ and denominators at most two. Ergo, the Fourier coefficients have polynomially bounded numerators and constant denominators, establishing strong $\mathsf{NP}$-hardness.
\end{proof}

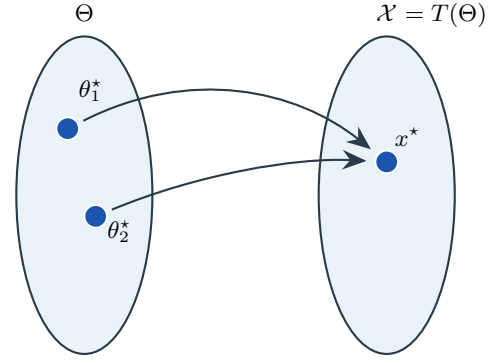
\begin{figure}
    \centering
\begin{tikzpicture}[
    x=1cm,
    y=1cm,
    region/.style={
        draw=mapoutline,
        fill=mapfill,
        line width=0.8pt
    },
    point/.style={
        circle,
        fill=mappoint,
        draw=white,
        line width=0.7pt,
        inner sep=0pt,
        minimum size=8.5pt
    },
    mapping/.style={
        draw=mapoutline,
        line width=0.8pt,
        -{Stealth[length=3.2mm,width=2.6mm]},
        shorten <=2pt,
        shorten >=3pt
    }
]

\draw[region] (0,0)
    ellipse[x radius=0.9cm, y radius=2.1cm];

\draw[region] (4,0)
    ellipse[x radius=0.9cm, y radius=2.1cm];

\node[anchor=south, inner sep=0pt]
    at (0,2.30) {$\Theta$};

\node[anchor=south, inner sep=0pt]
    (Xlabel) at (4,2.30) {$\mathcal{X}$};

\node[anchor=base west, inner sep=0pt]
    at (Xlabel.base east) {${}=T(\Theta)$};

\node[point] (theta1) at (-0.22, 0.88) {};
\node[point] (theta2) at ( 0.15,-0.28) {};
\node[point] (xstar)  at ( 4.00, 0.44) {};

\draw[mapping]
    (theta1)
    .. controls (1.40,1.70) and (2.80,1.40) ..
    (xstar);

\draw[mapping]
    (theta2)
    .. controls (1.50,0.28) and (2.80,0.50) ..
    (xstar);

\node[anchor=south west, inner sep=0pt]
    at (-0.08,1.24) {$\theta_1^{\star}$};

\node[anchor=north west, inner sep=0pt]
    at (0.30,-0.33) {$\theta_2^{\star}$};

\node[anchor=south west, inner sep=0pt]
    at (4.10,0.65) {$x^{\star}$};

\end{tikzpicture}
    \caption{Illustration of the mapping between local minimizers under $T$ as stated in \cref{lemma:local-homeomorphism}, where $x^\star=T(\theta^\star_1)=T(\theta^\star_2)$.}
    \label{fig:local-homeomorphism}
\end{figure}

\section{Reduction to VQA training}

We now prove \cref{thm:vqa-main} by giving a polynomial-time reduction from the Hermitian trigonometric polynomial local-minimization problem in \cref{thm:htp-nphard} to VQA training. We construct a polynomial-size parameterized quantum circuit (PQC) and an observable whose training objective $F(\theta)$ is exactly the hard Hermitian trigonometric polynomial $f(\theta)$.

\begin{proof}[Proof of \cref{thm:vqa-main}]

Fix $p \geq 1$, and let the Hermitian trigonometric polynomial $f(\theta)$ in Eq.~\eqref{eq:HTPChangeOfVariables} be a hard instance of the Hermitian trigonometric polynomial local-minimization problem.
In order to reproduce the products $T(\theta_i)T(\theta_j)$ appearing in $f$, we start by initializing $2n$ qubits in the state $\lvert 0^{2n}\rangle$.
Each parameter $\theta_i$ is applied to the $i$-th and $(n+i)$-th qubits.
In particular, we define the PQC
\begin{align}
U(\theta)
= \prod_{i = 1}^{n} R_{x,i}(\theta_i)R_{x,n+i}(\theta_i),
\end{align}
where $R_{x,j}(\phi) = \mathrm{e}^{-\mathrm{i}\phi X_j/2}$, for $\phi \in \mathbb{R}$, denotes an $x$-rotation acting on qubit $j$.

We next define the Pauli-sum observable
\begin{align}
H
= \frac{1}{4}\sum_{i = 1}^{n}\sum_{j = 1}^{n}M_{ij}(I-Z_i)(I-Z_{n+j}).
\end{align}
Then, since $U(\theta)\lvert 0^{2n}\rangle$ is a product state and $I-Z_i$ and $I-Z_{n+j}$ act on distinct qubits, their expectation value factorizes.
In particular, we have
\begin{align}
F(\theta)
&= \langle 0^{2n}\rvert U(\theta)^{\dagger}HU(\theta)\lvert 0^{2n}\rangle \nonumber\\
&= \frac{1}{4}\sum_{i = 1}^{n}\sum_{j = 1}^{n}M_{ij}S_i(\theta_i)S_{n+j}(\theta_j),
\end{align}
where $S_i(\phi)$ denotes the expectation value of $I-Z_i$ in the prepared state $R_{x,i}(\phi)\lvert 0\rangle$, i.e.,
\begin{align}
S_i(\phi)
= \langle 0\rvert R_{x,i}(\phi)^{\dagger}(I-Z_i)R_{x,i}(\phi)\lvert 0\rangle.
\end{align}

To conclude the proof, it remains to show that $S_i(\phi)/2 = T(\phi)$ for every qubit $i$ and every angle $\phi$.
Indeed, since $(I-Z)/2 = \lvert 1\rangle\!\langle 1\rvert$, we have
\begin{align}
\tfrac{1}{2}S_i(\phi)
&= \left\lvert\langle 1\rvert R_{x,i}(\phi)\lvert 0\rangle\right\rvert^2 = \sin^2(\phi/2) = \tfrac 12 (1-\cos(\phi)) \nonumber\\
&= T(\phi),
\qquad \forall \phi \in [0,2\pi],
\end{align}
and therefore
\begin{align}
F(\theta)
= \sum_{i = 1}^{n}\sum_{j = 1}^{n}M_{ij}T(\theta_i)T(\theta_j)
= f(\theta).
\end{align}
This explicit expression allows efficient classical evaluation of the objective to any desired precision.

Finally, since the circuit uses only $2n$ qubits and $2n$ gates, while $H$ has $\mathcal{O}(n^2)$ Pauli terms, the construction can be carried out in polynomial time. Since $F(\theta) = f(\theta)$, any algorithm that approximates a local minimizer of the VQA training objective $F(\theta)$ within $\ell_p$-distance strictly less than $\pi/2$ also approximates a local minimizer of the Hermitian trigonometric polynomial $f(\theta)$ within the same distance. Our result therefore follows from \cref{thm:htp-nphard}.
\end{proof}

\section{Concluding remarks}

We have shown that approximating the location of a local minimizer of a VQA training objective is strongly $\mathsf{NP}$-hard, even for objectives that admit efficient classical evaluation. Thus, replacing global optimization with local optimization does not, in general, circumvent the worst-case computational hardness of VQA training. Notably, this hardness already arises for circuits consisting solely of single-qubit rotations with shared parameters and a commuting diagonal two-local observable.

Despite this worst-case hardness, parameter optimization may be tractable for particular VQA families \cite{you2022convergence,v2026near-optimal,girardi2025trained}. For example, exploiting the structure of level-1 quantum approximate optimization algorithm (QAOA) objectives enables efficient parameter tuning for Ising models under suitable assumptions on the weights \cite{v2026near-optimal}. This raises the question of which restrictions on observables and circuit architectures distinguish such tractable families from those for which even approximating a local minimizer is hard. We leave a systematic characterization of this boundary to future work.

\begin{acknowledgments}
D.E.K.\ is supported by the National Research Foundation, Singapore through the National Quantum Office, hosted in A*STAR, under the Advanced Quantum Algorithms and Solutions (AQAS) Funding Initiative (S25Q9DA001). A.V.\ is supported by the MOE Tier 2 Grant (MOE-T2EP20223-0018), the CQT++ Core Research Funding Grant (SUTD) (RS-NRCQT-00002), and partially by Project MIS 5154714 of the National Recovery and Resilience Plan, Greece 2.0, funded by the
European Union under the NextGenerationEU Program.
\end{acknowledgments}

\section*{Data Availability}
No data were created or analyzed in this study.

\bibliography{localvqa}

\end{document}